\documentclass[journal]{IEEEtran}

\usepackage[usenames]{color}
\usepackage{amssymb}
\usepackage{graphicx}
\usepackage{amscd}

\usepackage{graphics,amsmath,amssymb,relsize}
\usepackage{bm}
\usepackage{amsthm}
\usepackage{amsfonts}
\usepackage{latexsym}
\usepackage{epsf}

\newtheorem{theorem}{Theorem}[section]
\newtheorem{corollary}[theorem]{Corollary}
\newtheorem{lemma}[theorem]{Lemma}

\newtheorem{defin}[theorem]{Definition}
\newenvironment{definition}{\begin{defin}\normalfont\quad}{\end{defin}}
\newtheorem{examp}[theorem]{Example}

\newtheorem{rema}[theorem]{Remark}
\newtheorem{prob}[theorem]{Problem}

\newenvironment{problem}{\begin{prob}\normalfont}{\end{prob}}

\numberwithin{equation}{section}

\newcommand{\bt}{\begin{thm}}
\newcommand{\et}{\end{thm}}
\newcommand{\bp}{\begin{proof}}
\newcommand{\ep}{\end{proof}}
\newcommand{\bprop}{\begin{prop}}
\newcommand{\eprop}{\end{prop}}
\newcommand{\bl}{\begin{lemma}}
\newcommand{\el}{\end{lemma}}
\newcommand{\bc}{\begin{corollary}}
\newcommand{\ec}{\end{corollary}}
\newcommand{\Z}{\mathbb{Z}}

\newcommand{\be}{\begin{enumerate}}
\newcommand{\ee}{\end{enumerate}}

\newcommand{\OMIT}[1]{}

\ifCLASSINFOpdf
  
\else
  
\fi

\begin{document}

\title{Explicit Formulas for the Weight Enumerators \\ of Some Classes of Deletion Correcting Codes}

\author{Khodakhast~Bibak and~Olgica~Milenkovic,~\IEEEmembership{Fellow,~IEEE}%

\thanks{An extended abstract of this paper was presented at ISIT 2018, Colorado, USA~\cite{BM}. K.~Bibak is with the Department of Computer Science and Software Engineering, Miami University, Oxford, Ohio, 45056, USA. Email: bibakk@miamioh.edu. O.~Milenkovic is with the Coordinated Science Laboratory, University of Illinois at Urbana-Champaign, Urbana, Illinois, 61801, USA. Email: milenkov@illinois.edu.}}

\maketitle

\begin{abstract}
We introduce a general class of codes which includes several well-known classes of deletion/insertion correcting codes as special cases. For example, the Helberg code, the Levenshtein code, the Varshamov--Tenengolts code, and most variants of these codes including most of those which have been recently used in studying DNA-based data storage systems are all special cases of our code. Then, using a number theoretic method, we give an explicit formula for the weight enumerator of our code which in turn gives explicit formulas for the weight enumerators and so the sizes of all the aforementioned codes. We also obtain the size of the Shifted Varshamov--Tenengolts code. Another application which automatically follows from our result is an explicit formula for the number of binary solutions of an \textit{arbitrary} linear congruence which, to the best of our knowledge, is the first result of its kind in the literature and might be also of independent interest. Our general result might have more applications/implications in information theory, computer science, and mathematics.
\end{abstract}

\begin{IEEEkeywords}
Binary solution, BLCC, weight enumerator, deletion correcting code, linear congruence. 
\end{IEEEkeywords}

\IEEEpeerreviewmaketitle

\section{Introduction}\label{Sec_1}

\IEEEPARstart{D}{eletions} or insertions can occur in many systems; for example, they can occur in some communication and storage channels, in biological sequences, etc. Therefore, studying deletion/insertion correcting codes may lead to important insight into genetic processes and into many communication problems. Deletion correcting codes have been the subject of intense research for more than fifty years~\cite{MEBT, MITZ, SLO}, with recent results settling long standing open problems regarding constructions of multiple deletion correcting codes with low redundancy~\cite{BGH17,BGZ17}. Nevertheless, our understanding about these codes and channels with this type of errors is still very limited and many open problems in the area remain, especially when considering constructions of deletion correcting codes that satisfy additional constraints, such as weight or parity constraints. Examples include codes in the Damerau distance~\cite{GYM}, based on single deletion correcting codes with even weight, and Shifted Varshamov--Tenengolts codes~\cite{SWGY} used for burst deletion correction. In such settings, one important question is to determine the weight enumerators of the component deletion correcting codes in order to estimate the size~\cite{CUKI, KUKI} of the weight-constrained deletion correcting codes. The component deletion correcting code is frequently defined in terms of a linear congruence for which the number of solutions of some fixed weight determines the size of the constrained code.

Here, we introduce a general class of codes which includes several well-known classes of deletion/insertion correcting codes as special cases. Then, using a number theoretic method, we give an explicit formula for the weight enumerator of our code which in turn gives explicit formulas for the weight enumerators and for the sizes of the aforedescribed codes (see also \cite{CUKI, KUKI} for some general upper bounds for the size of deletion correcting codes). Our initial motivation for studying this problem comes from number theory, and pertains to a possible $q$-ary generalization of Lehmer's Theorem (see Section~\ref{Sec_2}). 

Before we proceed with our technical exposition, we review some well-known classes of deletion correcting codes.

Throughout the paper, we let $\Z_n = \lbrace 0, \ldots, n-1 \rbrace$. Varshamov and Tenengolts~\cite{VATE} in 1965 introduced an important class of codes, known as the Varshamov--Tenengolts codes (henceforth, VT-codes), and proved that these codes are capable of correcting single asymmetric errors on a $Z$-channel.  

\begin{definition}
Let $n$ be a positive integer and $b\in \Z_{n+1}$. The \textit{Varshamov--Tenengolts code} $VT_b(n)$ is the set of all binary $n$-tuples $\langle s_1,\ldots,s_n \rangle$ such that 
$$
\sum_{i=1}^{n}is_i \equiv b \pmod{n+1}.
$$ 
\end{definition}

A generalization of VT-codes to Abelian groups where the code length is one less than the order of the group was proposed by Constantin and Rao \cite{CORA}; the size and weight distribution of the latter codes were studied in \cite{DEPI, HEKL, KL, MCRO}. Despite the fact that the VT codes can correct only a single deletion~\cite{LEV1}, the codes and their variants have found many applications, including DNA-based data storage~\cite{GYM, LSWY} and distributed message synchronization~\cite{RAMJI,LARA}.

Levenshtein~\cite{LEV1} proved that any code that can correct $s$ deletions (or $s$ insertions) can also correct a total of $s$ deletions and insertions. In the same paper, he also proposed the following important generalization of VT codes.

\begin{definition}
Let $n$, $k$ be positive integers and $b\in \Z_n$. The \textit{Levenshtein code} $L_b(k,n)$ is the set of all binary $k$-tuples $\langle s_1,\ldots,s_k \rangle$ such that 
$$
\sum_{i=1}^{k}is_i \equiv b \pmod{n}.
$$ 
\end{definition}

By giving an elegant decoding algorithm, Levenshtein \cite{LEV1} showed that if $n\geq k+1$, then the code $L_b(k,n)$ can correct a single deletion (and consequently, can correct a single insertion). Furthermore, Levenshtein~\cite{LEV1} proved that if $n\geq 2k$ then the code $L_b(k,n)$ can correct either a single deletion/insertion error or a single substitution error. The Levenshtein code has found many interesting applications and is considered to be one of the most important examples of deletion/insertion correcting codes.

Motivated by applications in burst of deletion correction, a variant of the Levenshtein code was introduced in \cite{SWGY} under the name of Shifted Varshamov--Tenengolts codes. Gabrys et al.~\cite{GYM} used Shifted VT-codes to construct codes in the Damerau distance. Shifted VT-codes combine a linear congruence constraint with a parity constraint, as stated in the next definition.

\begin{definition}
Let $n$, $k$ be positive integers, $b\in \Z_n$, and $r \in \{0,1\}$. The \textit{Shifted Varshamov--Tenengolts code} $SVT_{b,r}(k,n)$ is the set of all binary $k$-tuples $\langle s_1,\ldots,s_k \rangle$ such that 
$$
\sum_{i=1}^{k}is_i \equiv b \pmod{n}, \;\;\;\;\; \sum_{i=1}^{k}s_i \equiv r \pmod{2}.  
$$ 
\end{definition}

The reason why these codes are called ``shifted" is that they can correct a single deletion where the location of the deleted bit is known to be within certain consecutive positions. A variation of the Shifted VT-codes appeared in~\cite{CKK1, CKK2}.

Helberg and Ferreira~\cite{HEFE} introduced a generalization of the Levenshtein code, referred to as the Helberg code, by replacing the coefficients (weights) $i$ with modified versions of the Fibonacci numbers.

\begin{definition}
Let $s$, $k$ be positive integers. The \textit{Helberg code} $H_b(k,s)$ is the set of all binary $k$-tuples $\langle s_1,\ldots,s_k \rangle$ such that 
$$
\sum_{i=1}^{k}v_i s_i \equiv b \pmod{n},
$$ 
where $v_i=0$, for $i\leq 0$, $v_i=1+\sum_{j=1}^s v_{i-j}$, for $i\geq 1$, $n=v_{k+1}$, and $b\in \Z_n$. Note that the multipliers $v_i$ depend on $s$, and $n$ depends on both $s$ and $k$.
\end{definition}

Clearly, the Helberg code with $s=1$ coincides with the VT code. Helberg and Ferreira~\cite{HEFE} gave numerical values for the maximum cardinality of this code for some special parameter choices. Abdel-Ghaffar et al. ~\cite{APFC} proved that the Helberg code can correct multiple deletion/insertion errors (see also~\cite{HAG} for a short proof of this result). Furthermore, multiple deletion correcting codes over nonbinary alphabets generalizing the Helberg code were recently proposed by Le and Nguyen~\cite{LENG}. The Helberg code constraint was combined with the parity constraint of Shifted VT-codes for the purpose of devising special types of DNA-based data storage codes in~\cite{GYM}. 

We now introduce our general code family which includes the above codes as special cases.

\begin{definition}
Let $n$, $k$ be positive integers, $a_1,\ldots,a_k\in \Z$, and $b\in \Z_n$. We define the \textit{Binary Linear Congruence Code} (BLCC) ${\cal C}$ as the set of all binary $k$-tuples $\langle c_1,\ldots,c_k \rangle$ such that 
$$
a_1c_1+\cdots +a_kc_k\equiv b \pmod{n}.
$$ 
\end{definition} 

The \textit{Hamming weight} of a string $\mathbf{s}$ over an alphabet, denoted by $w(\mathbf{s})$, is the number of non-zero symbols in $\mathbf{s}$. Equivalently, the Hamming weight of a string is the Hamming distance between that string and the all-zero string of the same length. The weight enumerator of a code is defined as follows.

\begin{definition}
Let $k$ be a positive integer, $\mathbb{F}$ be a finite field, and let $C\subseteq \mathbb{F}^k$. Then the \textit{weight enumerator} of the code $C$ is defined as
$$
W_C(z)= \mathlarger{\sum}_{\mathbf{c} \in C}z^{w(\mathbf{c})}= \mathlarger{\sum}_{t=0}^{k}N_tz^t,
$$
where $w(\mathbf{c})$ is the Hamming weight of $\mathbf{c}$, and $N_t$ is the number of codewords in $C$ of Hamming weight $t$. Also, the \textit{homogeneous weight enumerator} of the code $C$ is defined as
$$
W_C(x,y)= y^kW_C\left(\frac{x}{y}\right)= \mathlarger{\sum}_{t=0}^{k}N_tx^ty^{k-t}.
$$
Clearly, by setting $z=1$ in the weight enumerator (or $x=y=1$ in the homogeneous weight enumerator) we obtain the size of code $C$. 
\end{definition}

What can we say about the size, or more generally, about the weight enumerator of the Binary Linear Congruence Code (BLCC) ${\cal C}$? In the next section, we review linear congruences, exponential sums and in particular, Ramanujan sums. Then, in Section~\ref{Sec_3}, we give an explicit formula for the weight enumerator of ${\cal C}$. In Section~\ref{Sec_4}, we derive explicit formulas for the weight enumerators and for the sizes of the previously described deletion correcting codes. We also obtain a formula for the size of the Shifted Varshamov--Tenengolts codes.

\section{Linear congruences and Ramanujan sums}\label{Sec_2}

Let $a_1,\ldots,a_k,b,n\in \Z$, $n\geq 1$. Throughout the paper, an ordered $k$-tuple of integers is denoted by $\langle a_1,\ldots,a_k\rangle$. Also, by $\mathbf{x}\cdot\mathbf{y}$ we mean the scalar product of the vectors $\mathbf{x}$ and $\mathbf{y}$. A linear congruence in $k$ unknowns $x_1,\ldots,x_k$ is of the form
\begin{align} \label{cong form}
a_1x_1+\cdots +a_kx_k\equiv b \pmod{n}.
\end{align}

A solution of (\ref{cong form}) is an ordered $k$-tuple of integers $\mathbf{x}=\langle x_1,\ldots,x_k \rangle \in \mathbb{Z}_n^k$ that satisfies (\ref{cong form}). The following result, proved by Lehmer~\cite{LEH2}, gives the number of solutions of the above linear congruence.

\begin{theorem}\label{Prop: lin cong}
Let $a_1,\ldots,a_k,b,n\in \Z$, $n\geq 1$. The linear congruence $a_1x_1+\cdots +a_kx_k\equiv b \pmod{n}$ has a solution $\langle x_1,\ldots,x_k \rangle \in \Z_{n}^k$ if and only if $\ell \mid b$, where
$\ell=\gcd(a_1, \ldots, a_k, n)$. Furthermore, if this condition is satisfied, then there are $\ell n^{k-1}$ solutions.
\end{theorem}

Lehmer's Theorem and its variants have been studied extensively and have found intriguing applications in several areas of mathematics, computer science, and physics (see \cite{BKS2, BKS4, BKS7, BKSTT, BKSTT2, COH0, JAWILL} and the references therein).

Now, we pose the following problem that asks for a $q$-ary generalization of Lehmer's Theorem:

\begin{problem}\label{Prob: lin cong}
Let $a_1,\ldots,a_k,b,n,q\in \Z$, $n,q\geq 1$, and $q\leq n$. Give an explicit formula for the number of solutions of the linear congruence $a_1x_1+\cdots +a_kx_k\equiv b \pmod{n}$ with $\langle x_1,\ldots,x_k \rangle \in \Z_{q}^k$.
\end{problem}

Note that we have only changed $\langle x_1,\ldots,x_k \rangle \in \Z_{n}^k$ to $\langle x_1,\ldots,x_k \rangle \in \Z_{q}^k$. For example, when $q=2$, the problem is asking for an explicit formula for the number of binary solutions of an \textit{arbitrary} linear congruence. This is a very natural problem and might lead to interesting applications. In Section~\ref{Sec_3}, we solve the binary version of the above problem as an immediate consequence of our main result.

\begin{rema}
A solution to Problem~\ref{Prob: lin cong} automatically gives the size of a multiple insertion/deletion correcting code recently proposed by Le and Nguyen \cite{LENG} which generalize the Helberg code.
\end{rema} 

Next, we review some properties of exponential sums and in particular, Ramanujan sums. Throughout the paper, we let $e(x)=\exp(2\pi ix)$ denote the complex exponential with period $1$. 

\begin{lemma}\label{lem: sine}
Let $n$ be a positive integer and $x$ be a real number. Then we have
\begin{equation}\label{lem: sineF}
\sum_{m=1}^{n}e(mx) = \begin{cases}
n, \ & \text{ if $x \in \Z$}, \\  \frac{\sin (nx \pi)}{\sin (x \pi)}e\left(\frac{(n+1)x}{2}\right), \ & \text{ if $x \in \mathbb{R} \setminus \Z$}.
\end{cases}
\end{equation}
\end{lemma}

\begin{proof}
When $x \in \Z$ the result is clear because in this case $e(x)=e(mx)=1$. So, we let $x \in \mathbb{R} \setminus \Z$. Since $e(x)\not=1$, summing the geometric progression gives
\begin{align*}
\sum_{m=1}^{n}e(mx) &= \frac{e(x)\left(1-e(nx)\right)}{1-e(x)}\\
&= \frac{e\left(\frac{x}{2}\right)\left(1-e(nx)\right)}{e\left(\frac{-x}{2}\right)\left(1-e(x)\right)}\\
&= \frac{e\left(\frac{x}{2}\right)-e\left(\frac{(2n+1)x}{2}\right)}{e\left(\frac{-x}{2}\right)-e\left(\frac{x}{2}\right)}\\
&= \frac{e\left(\frac{-nx}{2}\right)-e\left(\frac{nx}{2}\right)}{e\left(\frac{-x}{2}\right)-e\left(\frac{x}{2}\right)}e\left(\frac{(n+1)x}{2}\right)\\
&= \frac{-2i\sin (nx \pi)}{-2i\sin (x \pi)}e\left(\frac{(n+1)x}{2}\right)\\
&= \frac{\sin (nx \pi)}{\sin (x \pi)}e\left(\frac{(n+1)x}{2}\right). 
\end{align*}
\end{proof}

For integers $m$ and $n$ with $n \geq 1$ the quantity 
\begin{align}\label{def1}
c_n(m) = \mathlarger{\sum}_{\substack{j=1 \\ \gcd(j,n)=1}}^{n}
e\!\left(\frac{jm}{n}\right)
\end{align}
is called a {\it Ramanujan sum}. It is the sum of the $m$-th powers of the primitive $n$-th roots of unity, and is also denoted by $c(m,n)$ in the literature. Even though the Ramanujan sum $c_n(m)$ is defined as a sum of some complex numbers, it is integer-valued (see Theorem~\ref{thm:Ram Mob} below). From (\ref{def1}), it is clear that $c_n(-m) = c_n(m)$. 

Ramanujan sums and some of their properties were certainly known before Ramanujan's paper \cite{RAM}, as Ramanujan himself declared \cite{RAM}; nonetheless, probably the reason that these sums bear Ramanujan's name is that ``Ramanujan was the first to appreciate the importance of the sum and to use it systematically", according to Hardy (see, \cite{FGK} for a discussion about this).

Ramanujan sums have important applications in additive number theory, for example, in the context of the Hardy-Littlewood circle method, Waring's problem, and sieve theory (see, e.g., \cite{MOVA, NAT, VAU} and the references therein). As a major result in this direction, one can mention Vinogradov's theorem (in its proof, Ramanujan sums play a key role) stating that every sufficiently large odd integer is the sum of three primes, and so every sufficiently large even integer is the sum of four primes (see, e.g., \cite[Chapter 8]{NAT}). Ramanujan sums have also interesting applications in cryptography \cite{BKSTT2, SCS}, coding theory \cite{BKS7, GIN}, combinatorics \cite{BKSTT, MENE}, graph theory \cite{DR, MS}, signal processing \cite{VAI1, VAI2}, and physics \cite{BKS2, PMS}.

Clearly, $c_n(0)=\varphi (n)$, where $\varphi (n)$ is {\it Euler's totient function}. Also, by Theorem~\ref{thm:Ram Mob} (see below), $c_n(1)=\mu (n)$, where $\mu (n)$ is the {\it M\"{o}bius function} defined by
\begin{align}\label{def2}
 \mu (n)&=
  \begin{cases}
    1, & \text{if $n=1$,}\\
    0, & \text{if $n$ is not square-free,}\\
    (-1)^{\kappa}, & \text{if $n$ is the product of $\kappa$ distinct primes}.
  \end{cases}
\end{align}

The following theorem, attributed to Kluyver~\cite{KLU}, gives an explicit formula for $c_n(m)$:

\begin{theorem} \label{thm:Ram Mob}
For integers $m$ and $n$, with $n \geq 1$,
\begin{align}\label{for:Ram Mob}
c_n(m) = \mathlarger{\sum}_{d\, \mid\, \gcd(m,n)} \mu
\!\left(\frac{n}{d}\right)d.
\end{align}
\end{theorem}

Thus, $c_n(m)$ can be easily computed provided $n$ can be factored efficiently. One should compare (\ref{for:Ram Mob}) with the formula
\begin{align}\label{for:Phi Mob}
\varphi (n) = \mathlarger{\sum}_{d\, \mid\, n} \mu \left(\frac{n}{d}\right)d.
\end{align}

\section{Weight enumerator of the Binary Linear Congruence Code}\label{Sec_3}

Using a simple number theoretic argument, we give an explicit formula for the weight enumerator (and the size) of the Binary Linear Congruence Code (BLCC) ${\cal C}$. Another result which automatically follows from our result is an explicit formula for the number of binary solutions of an \textit{arbitrary} linear congruence which, to the best of our knowledge, is the first result of its kind in the literature and may be of independent interest. 

The following lemma is useful for proving our main result. 

\begin{lemma}\label{lem: binary}
Let $n$, $k$ be positive integers. For any $k$-tuple $\mathbf{m}=\langle m_1,\ldots,m_k \rangle \in \mathbb{C}^k$, we have 
\begin{align}\label{lem: binaryF}
\mathlarger{\prod}_{j=1}^{k}\left(1+ze\left(\frac{m_j}{n}\right)\right)=\mathlarger{\sum}_{\mathbf{d} \in \{0,1\}^k}e\left(\frac{\mathbf{d}\cdot\mathbf{m}}{n}\right)z^{w(\mathbf{d})}.
\end{align} 
\end{lemma}

\begin{proof}
Expand the left-hand side of (\ref{lem: binaryF}) and note that $e(x)e(y)=e(x+y)$. 
\end{proof}

Now we are ready to state and prove our main result. 

\begin{theorem}\label{BLCCWE}
Let $n$, $k$ be positive integers, $a_1,\ldots,a_k\in \Z$, and $b\in \Z_n$. The weight enumerator of the Binary Linear Congruence Code \textnormal{(BLCC)} ${\cal C}$ is
\begin{align}\label{BLCCWEF}
W_{\cal C}(z)= \frac{1}{n}\mathlarger{\sum}_{m=1}^{n}e\left(\frac{-bm}{n}\right)\mathlarger{\prod}_{j=1}^{k}\left(1+ze\left(\frac{a_jm}{n}\right)\right).
\end{align}
\end{theorem}

\begin{proof}
By Lemma~\ref{lem: binary}, for any $k$-tuple $\mathbf{m}=\langle m_1,\ldots,m_k \rangle \in \mathbb{C}^k,$ we have 
\begin{align*}
\mathlarger{\prod}_{j=1}^{k}\left(1+ze\left(\frac{m_j}{n}\right)\right)=\mathlarger{\sum}_{\mathbf{d}=\langle d_1,\ldots,d_k \rangle \in \{0,1\}^k}e\left(\frac{\mathbf{d}\cdot\mathbf{m}}{n}\right)z^{w(\mathbf{d})}.
\end{align*}
Let $\mathbf{y}=\langle y_1,\ldots,y_k \rangle \in \mathbb{Z}_n^k$ be a solution of the linear congruence $a_1x_1+\cdots +a_kx_k\equiv b \pmod{n}$. Then we have
\begin{align*}
\ & e\left(\frac{-(\mathbf{m}\cdot\mathbf{y})}{n}\right)\mathlarger{\prod}_{j=1}^{k}\left(1+ze\left(\frac{m_j}{n}\right)\right)\\
= \ & \mathlarger{\sum}_{\mathbf{d}=\langle d_1,\ldots,d_k \rangle \in \{0,1\}^k}e\left(\frac{\mathbf{d}\cdot\mathbf{m}-\mathbf{m}\cdot\mathbf{y}}{n}\right)z^{w(\mathbf{d})}.
\end{align*}
Let $\mathbf{a}=\langle a_1,\ldots,a_k \rangle$ and $M=\{\langle a_1m,\ldots,a_km \rangle : m=1,\ldots,n\}$. Note that since $\mathbf{y}=\langle y_1,\ldots,y_k \rangle \in \mathbb{Z}_n^k$ is a solution of the linear congruence $a_1x_1+\cdots +a_kx_k\equiv b \pmod{n}$, we get $a_1y_1+\cdots +a_ky_k=\alpha n +b$, for some $\alpha \in \mathbb{Z}$. Similarly, $a_1d_1+\cdots +a_kd_k=\beta n +b'$, for some $\beta \in \mathbb{Z}$ and $b' \in \mathbb{Z}_n$. 

Therefore,
\begin{align*}
\ & \mathlarger{\sum}_{\mathbf{m} \in M}e\left(\frac{-(\mathbf{m}\cdot\mathbf{y})}{n}\right)\mathlarger{\prod}_{j=1}^{k}\left(1+ze\left(\frac{m_j}{n}\right)\right)\\
= \ & \mathlarger{\sum}_{\mathbf{d} \in \{0,1\}^k}\left(\mathlarger{\sum}_{\mathbf{m} \in M} e\left(\frac{\mathbf{d}\cdot\mathbf{m}-\mathbf{m}\cdot\mathbf{y}}{n}\right)\right)z^{w(\mathbf{d})}.
\end{align*}
Thus, 
\begin{align*}
\ & \mathlarger{\sum}_{m=1}^{n}e\left(\frac{-m(\mathbf{a}\cdot\mathbf{y})}{n}\right)\mathlarger{\prod}_{j=1}^{k}\left(1+ze\left(\frac{a_jm}{n}\right)\right)
\end{align*}
\begin{align*}
= \ & \mathlarger{\sum}_{\mathbf{d}=\langle d_1,\ldots,d_k \rangle \in \{0,1\}^k}\left(\mathlarger{\sum}_{m=1}^{n} e\left(\frac{m(\mathbf{d}\cdot\mathbf{a}-\mathbf{a}\cdot\mathbf{y})}{n}\right)\right)z^{w(\mathbf{d})}\Longrightarrow \\
\ & \mathlarger{\sum}_{m=1}^{n}e\left(\frac{-m(\alpha n + b)}{n}\right)\mathlarger{\prod}_{j=1}^{k}\left(1+ze\left(\frac{a_jm}{n}\right)\right)\\ 
= \ & \mathlarger{\sum}_{\mathbf{d}=\langle d_1,\ldots,d_k \rangle \in \{0,1\}^k}\left(\mathlarger{\sum}_{m=1}^{n} e\left(\frac{m((\beta-\alpha)n+b'-b)}{n}\right)\right)z^{w(\mathbf{d})}.
\end{align*}
Thus,
\begin{align*}
\ & \mathlarger{\sum}_{m=1}^{n}e\left(\frac{-bm}{n}\right)\mathlarger{\prod}_{j=1}^{k}\left(1+ze\left(\frac{a_jm}{n}\right)\right)\\ 
= \ & \mathlarger{\sum}_{\mathbf{d}=\langle d_1,\ldots,d_k \rangle \in \{0,1\}^k}\left(\mathlarger{\sum}_{m=1}^{n} e\left(\frac{m(b'-b)}{n}\right)\right)z^{w(\mathbf{d})}\\
= \ & \mathlarger{\sum}_{\mathbf{d}=\langle d_1,\ldots,d_k \rangle \in {\cal C}}\left(\mathlarger{\sum}_{m=1}^{n} e\left(\frac{m(b'-b)}{n}\right)\right)z^{w(\mathbf{d})}\\
+ \ & \mathlarger{\sum}_{\mathbf{d}=\langle d_1,\ldots,d_k \rangle \in \{0,1\}^k \setminus {\cal C}}\left(\mathlarger{\sum}_{m=1}^{n} e\left(\frac{m(b'-b)}{n}\right)\right)z^{w(\mathbf{d})}.
\end{align*}
By Lemma~\ref{lem: sine},
\begin{equation*}
\mathlarger{\sum}_{m=1}^n e\!\left(\frac{m(b'-b)}{n}\right) = \begin{cases}
n, \ & \text{ if $n\mid b'-b$}, \\  0, \ & \text{ if $n\nmid b'-b$}.
\end{cases}
\end{equation*}
Note that if $\mathbf{d}=\langle d_1,\ldots,d_k \rangle \in {\cal C}$ then $b'=b$ (and so $n\mid b'-b$), and if $\mathbf{d}=\langle d_1,\ldots,d_k \rangle \in \{0,1\}^k \setminus {\cal C}$ then $b'\not=b$ (and so $n\nmid b'-b$ because $b',b \in \mathbb{Z}_n$). This implies that
\begin{align*}
\mathlarger{\sum}_{\mathbf{d} \in {\cal C}}\left(\mathlarger{\sum}_{m=1}^{n} e\left(\frac{m(b'-b)}{n}\right)\right)z^{w(\mathbf{d})}=n\mathlarger{\sum}_{\mathbf{d} \in {\cal C}}z^{w(\mathbf{d})},
\end{align*}
and
\begin{align*}
\mathlarger{\sum}_{\mathbf{d} \in \{0,1\}^k \setminus {\cal C}}\left(\mathlarger{\sum}_{m=1}^{n} e\left(\frac{m(b'-b)}{n}\right)\right)z^{w(\mathbf{d})}=0.
\end{align*}
Consequently,
\begin{align*}
W_{\cal C}(z)=\mathlarger{\sum}_{\mathbf{c} \in {\cal C}}z^{w(\mathbf{c})}= \frac{1}{n} \mathlarger{\sum}_{m=1}^{n}e\left(\frac{-bm}{n}\right)\mathlarger{\prod}_{j=1}^{k}\left(1+ze\left(\frac{a_jm}{n}\right)\right). 
\end{align*}
\end{proof}

Setting $z=1$ in (\ref{BLCCWEF}) gives the size of the Binary Linear Congruence Code (BLCC) ${\cal C}$. Equivalently, it solves Problem~\ref{Prob: lin cong} when $q=2$, that is, it gives an explicit formula for the number of binary solutions of an \textit{arbitrary} linear congruence.

\begin{corollary}\label{BLCCS}
Let $n$, $k$ be positive integers, $a_1,\ldots,a_k\in \Z$, and $b\in \Z_n$. The number of solutions of the linear congruence $a_1x_1+\cdots +a_kx_k\equiv b \pmod{n}$ in $\Z_{2}^k$ is
\begin{align}\label{BLCCSF}
W_{\cal C}(1)= \frac{2^k}{n}\mathlarger{\sum}_{m=1}^{n}e\left(\frac{\eta m}{n}\right)\mathlarger{\prod}_{j=1}^{k}\cos \left(\frac{\pi a_jm}{n}\right)\geq 0,
\end{align}
where $\eta=-b+\frac{1}{2}\sum_{j=1}^{k}a_j$. This implies that 
\begin{align}\label{BLCCSFB}
W_{\cal C}(1)\leq \frac{2^k}{n}\mathlarger{\sum}_{m=1}^{n}\mathlarger{\prod}_{j=1}^{k}\left|\cos \left(\frac{\pi a_jm}{n}\right)\right|.
\end{align}
\end{corollary}

\begin{proof}
We have
\begin{align*}
\ & W_{\cal C}(1)= \frac{1}{n}\mathlarger{\sum}_{m=1}^{n}e\left(\frac{-bm}{n}\right)\mathlarger{\prod}_{j=1}^{k}\left(1+e\left(\frac{a_jm}{n}\right)\right)\\
&= \frac{1}{n}\mathlarger{\sum}_{m=1}^{n}e\left(\frac{-bm}{n}\right)\mathlarger{\prod}_{j=1}^{k}e\left(\frac{a_jm}{2n}\right)\mathlarger{\prod}_{j=1}^{k}\left(e\left(\frac{-a_jm}{2n}\right)+e\left(\frac{a_jm}{2n}\right)\right)\\
&= \frac{1}{n}\mathlarger{\sum}_{m=1}^{n}e\left(\frac{-bm}{n}\right)e\left(\frac{m}{2n}\sum_{j=1}^{k}a_j\right)\mathlarger{\prod}_{j=1}^{k}2\cos \left(\frac{\pi a_jm}{n}\right)\\
&= \frac{2^k}{n}\mathlarger{\sum}_{m=1}^{n}e\left(\frac{\eta m}{n}\right)\mathlarger{\prod}_{j=1}^{k}\cos \left(\frac{\pi a_jm}{n}\right),    
\end{align*}
where $\eta=-b+\frac{1}{2}\sum_{j=1}^{k}a_j$. Consequently, we have
$$
W_{\cal C}(1)\leq \frac{2^k}{n}\mathlarger{\sum}_{m=1}^{n}\mathlarger{\prod}_{j=1}^{k}\left|\cos \left(\frac{\pi a_jm}{n}\right)\right|.
$$
\end{proof}

\begin{rema}
Recently, Gabrys et al.~\cite{GYM} proposed several variants of the Levenshtein code which are all special cases of our Binary Linear Congruence Code \textnormal{(BLCC)} ${\cal C}$. Theorem~\ref{BLCCWE} hence provides explicit formulas for the weight enumerators of such codes.
\end{rema}

\section{Weight enumerators of the aforementioned codes}\label{Sec_4}

Using Theorem~\ref{BLCCWE}, we now describe explicit formulas for the weight enumerators (and the sizes) of the Helberg code, the Levenshtein code, and the Varshamov--Tenengolts code. Note that the same approach may be used to derive the weight enumerators of most variants of these codes since they are special cases of Binary Linear Congruence Codes (BLCC) ${\cal C}$. In addition, we derive a formula for the size of the Shifted Varshamov--Tenengolts code.

\subsection{Weight enumerator of the Helberg code}

The Helberg code has the same structure as the Binary Linear Congruence Code (BLCC) ${\cal C}$ but with some additional restrictions on the coefficients and the modulus. So, Theorem~\ref{BLCCWE} immediately gives the following result.

\begin{theorem}\label{W-H}
The weight enumerator of the Helberg code $H_b(k,s)$ is
\begin{align}\label{W-H-F}
W_{H_b(k,s)}(z)= \frac{1}{n}\mathlarger{\sum}_{m=1}^{n}e\left(\frac{-bm}{n}\right)\mathlarger{\prod}_{j=1}^{k}\left(1+ze\left(\frac{v_jm}{n}\right)\right).
\end{align}
\end{theorem}

As the coefficients in the Helberg code are a modified version of the Fibonacci numbers, it may be possible to connect trigonometric sums as described above with the Fibonacci and Lucas numbers~\cite{BISH1}, and hence simplify (\ref{W-H-F}).

\begin{corollary}\label{S-H}
The size of the Helberg code $H_b(k,s)$ equals
\begin{align}\label{S-H-F}
W_{H_b(k,s)}(1)= \frac{2^k}{n}\mathlarger{\sum}_{m=1}^{n}e\left(\frac{\eta m}{n}\right)\mathlarger{\prod}_{j=1}^{k}\cos \left(\frac{\pi v_jm}{n}\right),
\end{align}
where $\eta=-b+\frac{1}{2}\sum_{j=1}^{k}v_j$. This implies that 
\begin{align}\label{S-H-FB}
W_{H_b(k,s)}(1)\leq \frac{2^k}{n}\mathlarger{\sum}_{m=1}^{n}\mathlarger{\prod}_{j=1}^{k}\left|\cos \left(\frac{\pi v_jm}{n}\right)\right|.
\end{align}
\end{corollary}

\subsection{Weight enumerator of the Levenshtein code}

Theorem~\ref{BLCCWE} also allows for deriving an explicit formula for the weight enumerator of the Levenshtein code.

\begin{theorem}\label{W-L}
The weight enumerator of the Levenshtein code $L_b(k,n)$ is
\begin{align}\label{W-L-F}
W_{L_b(k,n)}(z)= \frac{1}{n}\mathlarger{\sum}_{m=1}^{n}e\left(\frac{-bm}{n}\right)\mathlarger{\prod}_{j=1}^{k}\left(1+ze\left(\frac{jm}{n}\right)\right).
\end{align}
\end{theorem}

\begin{corollary}\label{S-L}
The size of the Levenshtein code $L_b(k,n)$ equals
\begin{align}\label{S-L-F}
W_{L_b(k,n)}(1)= \frac{2^k}{n}\mathlarger{\sum}_{m=1}^{n}e\left(\frac{\eta m}{n}\right)\mathlarger{\prod}_{j=1}^{k}\cos \left(\frac{\pi jm}{n}\right),
\end{align}
where $\eta=-b+\frac{1}{4}k(k+1)$. This implies that 
\begin{align}\label{S-L-FB}
W_{L_b(k,n)}(1)\leq \frac{2^k}{n}\mathlarger{\sum}_{m=1}^{n}\mathlarger{\prod}_{j=1}^{k}\left|\cos \left(\frac{\pi jm}{n}\right)\right|.
\end{align}
\end{corollary}

\subsection{The size of the Shifted Varshamov--Tenengolts code}

Next, using Theorem~\ref{W-L} once again, we give an explicit formula for the size of the Shifted Varshamov--Tenengolts code $SVT_{b,r}(k,n)$. Note that $SVT_{b,r}(k,n)$ represents the set of codewords in the Levenshtein code with even Hamming weight (when $r=0$) or with odd Hamming weight (when $r=1$).      

\begin{theorem}\label{S-SVT}
If $r=0$ then the size of the Shifted Varshamov--Tenengolts code $SVT_{b,0}(k,n)$ is
\begin{align}\label{S-SVT-F0}
|SVT_{b,0}(k,n)|= \frac{2^{k-1}}{n}\mathlarger{\sum}_{m=1}^{n}e\left(\frac{\eta m}{n}\right)\left(A+(-1)^{k}B\right),
\end{align}
and if $r=1$ then the size of $SVT_{b,1}(k,n)$ is
\begin{align}\label{S-SVT-F1}
|SVT_{b,1}(k,n)|= \frac{2^{k-1}}{n}\mathlarger{\sum}_{m=1}^{n}e\left(\frac{\eta m}{n}\right)\left(A+(-1)^{k+1}B\right),
\end{align}
where $\eta=-b+\frac{1}{4}k(k+1)$,
$$
A=\mathlarger{\prod}_{j=1}^{k}\cos \left(\frac{\pi jm}{n}\right) \; \text{and} \; B=\mathlarger{\prod}_{j=1}^{k}i \sin \left(\frac{\pi jm}{n}\right). 
$$
\end{theorem} 

\begin{proof}
To find the number of codewords in the Levenshtein code $L_b(k,n)$ with even Hamming weight (when $r=0$) and with odd Hamming weight (when $r=1$), we proceed as follows. If $r=0,$ then the size of $SVT_{b,0}(k,n)$ equals $\frac{1}{2}(W_{L_b(k,n)}(z)+W_{L_b(k,n)}(-z))|_{z=1}$, and if $r=1$, the size of $SVT_{b,1}(k,n)$ equals $\frac{1}{2}(W_{L_b(k,n)}(z)-W_{L_b(k,n)}(-z))|_{z=1}$. Invoking Theorem~\ref{W-L} proves the claimed result.
\end{proof}

\subsection{Weight enumerators of VT codes} 

Using Theorem~\ref{BLCCWE} we re-derive the formula for the weight enumerator of the Varshamov--Tenengolts code. Due to the special structure of the coefficients int he congruences, our formula simplifies significantly.

We start with the following lemma.

\begin{lemma}\label{lem: cyclo 1}
Let $n$ be a positive integer and $m$ be a non-negative integer. Then, we have 
$$
\mathlarger{\prod}_{j=1}^{n}\left(1-ze\left(\frac{jm}{n}\right)\right)=(1-z^{\frac{n}{d}})^d,
$$
where $d=\gcd(m,n)$.
\end{lemma}

\begin{proof}
It is well-known that (see, e.g., \cite[p. 167]{STAN})

$$
1-z^n=\mathlarger{\prod}_{j=1}^{n}\left(1-ze^{2\pi ij/n}\right).
$$
Letting $d=\gcd(m,n)$, we obtain

\begin{align*}
\mathlarger{\prod}_{j=1}^{n}\left(1-ze\left(\frac{jm}{n}\right)\right)&=\mathlarger{\prod}_{j=1}^{n}\left(1-ze^{2\pi ijm/n}\right) \\&= \mathlarger{\prod}_{j=1}^{n} \left(1-ze^{2\pi ij\frac{m/d}{n/d}}\right)\\&= \left(\mathlarger{\prod}_{j=1}^{n/d}\left(1-ze^{2\pi ij\frac{m/d}{n/d}}\right)\right)^d\\
&{\stackrel{\gcd(\frac{m}{d},\frac{n}{d})=1}{=}} \left(\mathlarger{\prod}_{j=1}^{n/d}\left(1-ze^{\frac{2\pi ij}{n/d}}\right)\right)^d\\&=(1-z^{\frac{n}{d}})^d.
\end{align*}
\end{proof}

\begin{theorem}\label{W-VT}
The weight enumerator of the VT code $VT_b(n)$ is
\begin{align}\label{W-VT-F}
W_{VT_b(n)}(z)= \frac{1}{(z+1)(n+1)}\mathlarger{\sum}_{d\, \mid \, n+1}c_{d}(b)(1-(-z)^d)^{\frac{n+1}{d}}.
\end{align}
\end{theorem}

\begin{proof}
Using Theorem~\ref{BLCCWE} we get
\begin{align*}
W_{VT_b(n)}(z) = \frac{1}{n+1}\mathlarger{\sum}_{m=1}^{n+1}e\left(\frac{-bm}{n+1}\right)\mathlarger{\prod}_{j=1}^{n}\left(1+ze\left(\frac{jm}{n+1}\right)\right).
\end{align*}
Therefore,
\begin{align*}
\ & (z+1)(n+1)W_{VT_b(n)}(z)\\
= \ & \mathlarger{\sum}_{m=1}^{n+1}e\left(\frac{-bm}{n+1}\right)\mathlarger{\prod}_{j=1}^{n+1}\left(1+ze\left(\frac{jm}{n+1}\right)\right)\\
= \ & \mathlarger{\sum}_{d\, \mid \, n+1}\mathlarger{\sum}_{\substack{m=1 \\ \gcd(m, n+1)=d}}^{n+1}e\left(\frac{-bm}{n+1}\right)\mathlarger{\prod}_{j=1}^{n+1}\left(1+ze\left(\frac{jm}{n+1}\right)\right).
\end{align*}
Now, using Lemma~\ref{lem: cyclo 1} we get
\begin{align*}
\ & (z+1)(n+1)W_{VT_b(n)}(z)\\
= \ & \mathlarger{\sum}_{d\, \mid \, n+1}\mathlarger{\sum}_{\substack{m=1 \\ \gcd(m, n+1)=d}}^{n+1}e\left(\frac{-bm}{n+1}\right)(1-(-z)^{\frac{n+1}{d}})^d\\ 
{\stackrel{m'=m/d}{=}} \ & \mathlarger{\sum}_{d\, \mid \, n+1}\mathlarger{\sum}_{\substack{m'=1 \\ \gcd(m', (n+1)/d)=1}}^{(n+1)/d}e\left(\frac{-bm'}{(n+1)/d}\right)(1-(-z)^{\frac{n+1}{d}})^d\\
= \ & \mathlarger{\sum}_{d\, \mid \, n+1}c_{(n+1)/d}(-b)(1-(-z)^{\frac{n+1}{d}})^d\\
= \ & \mathlarger{\sum}_{d\, \mid \, n+1}c_{(n+1)/d}(b)(1-(-z)^{\frac{n+1}{d}})^d\\
= \ & \mathlarger{\sum}_{d\, \mid \, n+1}c_{d}(b)(1-(-z)^d)^{\frac{n+1}{d}}. 
\end{align*}
\end{proof}

Based on Theorem~\ref{W-VT}, one can easily obtain the following explicit formula for the general term of the weight distribution of VT codes. This result was recently proved using a different method by Bibak et al.~\cite{BKS7} (for a related earlier result, see also~\cite{DOAN}).

\begin{theorem}\label{VT exa k1s}
The number of codewords with Hamming weight $t$ in the Varshamov--Tenengolts code $VT_b(n)$ equals
\begin{align} \label{VT exa k1s: for}
N_t(VT_b(n))=\frac{(-1)^t}{n+1}\mathlarger{\sum}_{d\, \mid \, n+1}(-1)^{\lfloor\frac{t}{d}\rfloor}c_{d}(b)\binom{\frac{n+1}{d}-1}{\lfloor\frac{t}{d}\rfloor}.
\end{align}
\end{theorem}

\begin{proof}
The proof reduces to using the binomial theorem to find the coefficient of $z^{t+1}$ in the sum of (\ref{W-VT-F}). 
\end{proof}

\begin{corollary}\label{VT exa tot}
The size of the VT code $VT_b(n)$ equals
\begin{align} \label{VT exa tot: for}
W_{VT_b(n)}(1)=\frac{1}{2(n+1)}\mathlarger{\sum}_{\substack{d\, \mid \, n+1 \\ d \; \textnormal{odd}}}c_{d}(b)2^{\frac{n+1}{d}}.
\end{align}
\end{corollary}

\begin{rema}\label{VTq exa tot}
Ginzburg~\cite{GIN} in 1967 proved the following explicit formula for the size $|VT_{b,q}(n)|$ of the $q$-ary, rather than binary, Varshamov--Tenengolts code $VT_{b,q}(n)$, where $q$ is an arbitrary positive integer:
\begin{align} \label{VTq exa tot: for}
|VT_{b,q}(n)|=\frac{1}{q(n+1)}\mathlarger{\sum}_{\substack{d\, \mid \, n+1 \\ \gcd(d,q)=1}}c_{d}(b)q^{\frac{n+1}{d}}.
\end{align}
Formula~(\ref{VTq exa tot: for}) (in fact, a more complicated version of it) was later rediscovered by Stanley and Yoder~\cite{STYO} in 1973. Formula~(\ref{VTq exa tot: for}) for the binary case $q=2$ was also rediscovered by Sloane~\cite{SLO} in 2002. Bibak et al.~\cite{BKS7} derived the binary case formula as a corollary of a general number theory problem.       
\end{rema}

\begin{rema}
Since for all integers $m$ and $n$ with $n \geq 1$ one has $c_n(m) \leq \varphi (n)$, from (\ref{VTq exa tot: for}) it is clear that the maximum number of codewords in the $q$-ary Varshamov--Tenengolts code $VT_{b,q}(n)$ is obtained for $b=0$, that is, 
$$
|VT_{0,q}(n)|=\frac{1}{q(n+1)}\mathlarger{\sum}_{\substack{d\, \mid \, n+1 \\ \gcd(d,q)=1}}\varphi (d)q^{\frac{n+1}{d}} \geq |VT_{b,q}(n)|,
$$
for all $b$. This result was originally proved by Ginzburg~\cite{GIN}.
\end{rema}

\begin{rema}
Setting $d=1$ in Formula~(\ref{VT exa tot: for}) gives the bound
$$
|VT_0(n)| \geq \frac{2^n}{n+1}.
$$
On the other hand, by a result of Levenshtein~\cite{LEV1}, the size of the largest single deletion correcting binary code of length $n$, where $n$ is sufficiently large, is roughly $\frac{2^n}{n}$. Therefore, as it is well-known, the VT-codes $VT_0(n)$, for sufficiently large $n$, are close to optimal.
\end{rema}

\section*{Acknowledgements}
The authors would like to thank the reviewers for helpful comments that improved the presentation of this paper. This work was supported in part by the Center for Science of Information (CSoI), an NSF Science and Technology Center, under grant agreement CCR-0939370, and by the NSF grant CCF1618366.

\begin{IEEEbiographynophoto}{Khodakhast Bibak}
is an Assistant Professor at the Department of Computer Science and Software Engineering at Miami University. He was a Postdoctoral Research Associate (2017--2018) in the Coordinated Science Laboratory at the University of Illinois at Urbana-Champaign. Khodakhast received his PhD in Computer Science from the University of Victoria in 2017, and his MMath degree in Combinatorics \& Optimization from the University of Waterloo in 2013. His research interests are mainly cryptography, information security, coding and information theory, and discrete mathematics.
\end{IEEEbiographynophoto}

\begin{IEEEbiographynophoto}{Olgica Milenkovic}
is a professor of Electrical and Computer Engineering at the University of Illinois, Urbana-Champaign (UIUC), and Research Professor at the Coordinated Science Laboratory. She obtained her Masters Degree in Mathematics in 2001 and PhD in Electrical Engineering in 2002, both from the University of Michigan, Ann Arbor. Prof. Milenkovic heads a group focused on addressing unique interdisciplinary research challenges spanning the areas of algorithm design and computing, bioinformatics, coding theory, machine learning and signal processing. Her scholarly contributions have been recognized by multiple awards, including the NSF Faculty Early Career Development (CAREER) Award, the DARPA Young Faculty Award, the Dean's Excellence in Research Award, and several best paper awards. In 2013, she was elected a UIUC Center for Advanced Study Associate and Willett Scholar while in 2015 she was elected a Distinguished Lecturer of the Information Theory Society. In 2018 she became an IEEE Fellow. She has served as Associate Editor of the IEEE Transactions of Communications, the IEEE Transactions on Signal Processing, the IEEE Transactions on Information Theory and the IEEE Transactions on Molecular, Biological and Multi-Scale Communications. In 2009, she was the Guest Editor in Chief of a special issue of the IEEE Transactions on Information Theory on Molecular Biology and Neuroscience.
\end{IEEEbiographynophoto}

\end{document}